\documentclass{article}
\usepackage{amssymb}
\usepackage{makeidx}
\usepackage{amsfonts}
\usepackage{amsmath}

\newtheorem{theorem}{Theorem}

\newtheorem{corollary}[theorem]{Corollary}

\newtheorem{definition}[theorem]{Definition}
\newtheorem{example}[theorem]{Example}

\newtheorem{proposition}[theorem]{Proposition}
\newtheorem{remark}[theorem]{Remark}

\newenvironment{proof}[1][Proof]{\noindent\textbf{#1.} }{\ \rule{0.5em}{0.5em}}
\input{tcilatex}
\begin{document}

\title{From quadratic Hamiltonians of polymomenta to abstract geometrical
Maxwell-like and Einstein-like equations}
\author{Alexandru Oan\u{a} and Mircea Neagu}
\date{}
\maketitle

\begin{abstract}
The aim of this paper is to create a large geometrical background on the
dual $1$-jet space $J^{1\ast }(\mathcal{T},M)$ for a multi-time Hamiltonian
approach of the electromagnetic and gravitational physical fields. Our
geometric-physical construction is achieved starting only from a given
quadratic Hamiltonian function 
\begin{equation*}
H=h_{ab}(t)g^{ij}(t,x)p_{i}^{a}p_{j}^{b}+U_{(a)}^{(i)}(t,x)p_{i}^{a}+%
\mathcal{F}(t,x)
\end{equation*}%
which naturally produces a canonical nonlinear connection $N$, a canonical
Cartan $N$-linear connection $C\Gamma \left( N\right) $ and their
corresponding local distinguished (d-) torsions and curvatures. In such a
context, we construct some geometrical electromagnetic-like and
gravitational-like field theories which are characterized by some natural
geometrical Maxwell-like and Einstein-like equations. Some abstract and
geometrical conservation laws for the multi-time Hamiltonian gravitational
physical field are also given.
\end{abstract}



\section{Distinguished Riemannian geometrization of metrical multi-time
Hamilton spaces}

Recently, the studies of Atanasiu and Neagu (see the papers \cite{Atan-Neag2}%
, \cite{Atan-Neag0} and \cite{Atan+Neag1}) initiated the new way of
distinguished Riemannian geometrization for Hamiltonians depending on
polymomenta, which represents in fact a natural "multi-time" extension of
the already classical Hamiltonian geometry on cotangent bundles (synthesized
in the Miron et al.'s book \cite{Miro+Hrim+Shim+Saba}). In what follows, we
expose the main geometrical ideas which characterize the distinguished
Riemannian geometrical approach of Hamiltonians depending on polymomenta
(see for details the Oan\u{a} and Neagu's papers \cite{Oana+Neag}, \cite%
{Oana+Neag-2}).

Let us consider that $h=\left( h_{ab}\left( t\right) \right) $ is a
semi-Riemannian metric on the "multi-time" (\textit{temporal}) manifold $%
\mathcal{T}^{m}$, where $m=\dim \mathcal{T}$. Let $g=\left(
g^{ij}(t^{c},x^{k},p_{k}^{c})\right) $ be a symmetric d-tensor on the dual $%
1 $-jet space $E^{\ast }=J^{1\ast }(\mathcal{T},M)$, which has the rank $%
n=\dim M$ and a constant signature. At the same time, let us consider a
smooth multi-time Hamiltonian function%
\begin{equation*}
E^{\ast }\ni (t^{a},x^{i},p_{i}^{a})\rightarrow H(t^{a},x^{i},p_{i}^{a})\in 
\mathbb{R},
\end{equation*}%
which yields the \textit{fundamental vertical metrical d-tensor}%
\begin{equation*}
G_{(a)(b)}^{(i)(j)}={\dfrac{1}{2}}{\frac{\partial ^{2}H}{\partial
p_{i}^{a}\partial p_{j}^{b}},}
\end{equation*}%
where $a,b=1,..,m$ and $i,j=1,...,n.$

\begin{definition}
A multi-time Hamiltonian function $H:E^{\ast }\rightarrow \mathbb{R},$
having the fundamental vertical metrical d-tensor of the form 
\begin{equation*}
G_{(a)(b)}^{(i)(j)}(t^{c},x^{k},p_{k}^{c})={\frac{1}{2}}{\frac{\partial ^{2}H%
}{\partial p_{i}^{a}\partial p_{j}^{b}}}%
=h_{ab}(t^{c})g^{ij}(t^{c},x^{k},p_{k}^{c}),
\end{equation*}%
is called a \textbf{Kronecker }$h$\textbf{-regular multi-time Hamiltonian
function}.
\end{definition}

In this context, we introduce the following important geometrical concept:

\begin{definition}
A pair $MH_{m}^{n}=(E^{\ast }=J^{1\ast }(\mathcal{T},M),H),$ where $m=\dim 
\mathcal{T}$ and $n=\dim M,$ consisting of the dual $1$-jet space and a
Kronecker $h$-regular multi-time Hamiltonian function $H:E^{\ast
}\rightarrow \mathbb{R},$ is called a \textbf{multi-time Hamilton space}.
\end{definition}

\begin{example}
Let us consider the Kronecker $h$-regular multi-time Hamiltonian function $%
H_{1}:E^{\ast }\rightarrow \mathbb{R},$ given by%
\begin{equation}
H_{1}=\frac{1}{4mc}h_{ab}(t)\varphi ^{ij}(x)p_{i}^{a}p_{j}^{b},  \label{G}
\end{equation}%
where $h_{ab}(t)$ ($\varphi _{ij}(x)$, respectively) is a semi-Riemannian
metric on the temporal (spatial, respectively) manifold $\mathcal{T}$ ($M$,
respectively) having the physical meaning of \textbf{gravitational potentials%
}, and $m$ and $c$ are the known constants from Theoretical Physics
representing the \textbf{mass of the test body} and the \textbf{speed of
light}. Then, the multi-time Hamilton space $\mathcal{G}MH_{m}^{n}=(E^{\ast
},H_{1})$ is called the \textbf{multi-time Hamilton space of the
gravitational field}.
\end{example}

\begin{example}
If we consider on $E^{\ast }$ a symmetric d-tensor field $g^{ij}(t,x)$,
having the rank $n$ and a constant signature, we can define the Kronecker $h$%
-regular multi-time Hamiltonian function $H_{2}:E^{\ast }\rightarrow \mathbb{%
R},$ by setting%
\begin{equation}
H_{2}=h_{ab}(t)g^{ij}(t,x)p_{i}^{a}p_{j}^{b}+U_{(a)}^{(i)}(t,x)p_{i}^{a}+%
\mathcal{F}(t,x),  \label{NED}
\end{equation}%
where $U_{(a)}^{(i)}(t,x)$ is a d-tensor field on $E^{\ast },$ and $\mathcal{%
F}(t,x)$ is a function on $E^{\ast }$. Then, the multi-time Hamilton space $%
\mathcal{NED}MH_{m}^{n}=(E^{\ast },H_{2})$ is called the \textbf{%
non-a\-u\-to\-no\-mous multi-time Hamilton space of electrodynamics}. The
dynamical character of the gravitational potentials $g_{ij}(t,x)$ (i.e., the
dependence on the temporal coordinates $t^{c}$) motivated us to use the word 
\textbf{"non-autonomous".}
\end{example}

An important role for the subsequent development of our distinguished
Riemannian geometrical theory for multi-time Hamilton spaces is
re\-pre\-sen\-ted by the following result (proved in the paper \cite%
{Atan-Neag2}):

\begin{theorem}
\label{thchar} If we have $m=\dim \mathcal{T}\geq 2,$ then the following
statements are equivalent:

\emph{(i)} $H$ is a Kronecker $h$-regular multi-time Hamiltonian function on 
$E^{\ast }$.

\emph{(ii)} The multi-time Hamiltonian function $H$ reduces to a multi-time
Hamiltonian function of non-autonomous electrodynamic type. In other words,
we have%
\begin{equation}
H=h_{ab}(t)g^{ij}(t,x)p_{i}^{a}p_{j}^{b}+U_{(a)}^{(i)}(t,x)p_{i}^{a}+%
\mathcal{F}(t,x).  \label{NEDTH}
\end{equation}
\end{theorem}

\begin{corollary}
The \textit{fundamental vertical metrical d-tensor of a }Kronecker $h$%
-regular multi-time Hamiltonian function $H$ has the form%
\begin{equation}
G_{(a)(b)}^{(i)(j)}={\frac{1}{2}}{\frac{\partial ^{2}H}{\partial
p_{i}^{a}\partial p_{j}^{b}}}=\left\{ 
\begin{array}{ll}
h_{11}(t)g^{ij}(t,x^{k},p_{k}^{1}), & m=\dim \mathcal{T}=1\medskip \\ 
h_{ab}(t^{c})g^{ij}(t^{c},x^{k}), & m=\dim \mathcal{T}\geq 2.%
\end{array}%
\right.  \label{FVDT}
\end{equation}
\end{corollary}

Following now the Miron's geometrical ideas from \cite{Miro+Hrim+Shim+Saba},
the paper \cite{Atan-Neag2} proves that any Kronecker $h$-regular multi-time
Hamiltonian function $H$ produces a natural nonlinear connection on the dual
1-jet space $E^{\ast }$, which depends only by the given Hamiltonian
function $H$:

\begin{theorem}
The pair of local functions $N=\left( \underset{1}{N}\text{{}}_{(i)b}^{(a)},%
\underset{2}{N}\text{{}}_{(i)j}^{(a)}\right) $ on $E^{\ast },$ where ($\chi
_{bc}^{a}$ are the Christoffel symbols of the semi-Riemannian temporal
metric $h_{ab}$)%
\begin{equation*}
\begin{array}{l}
\underset{1}{N}\text{{}}_{(i)b}^{(a)}=\chi _{bc}^{a}p_{i}^{c},\medskip \\ 
\underset{2}{N}\text{{}}_{(i)j}^{(a)}=\dfrac{h^{ab}}{4}\left[ \dfrac{%
\partial g_{ij}}{\partial x^{k}}\dfrac{\partial H}{\partial p_{k}^{b}}-%
\dfrac{\partial g_{ij}}{\partial p_{k}^{b}}\dfrac{\partial H}{\partial x^{k}}%
+g_{ik}\dfrac{\partial ^{2}H}{\partial x^{j}\partial p_{k}^{b}}+g_{jk}\dfrac{%
\partial ^{2}H}{\partial x^{i}\partial p_{k}^{b}}\right] ,%
\end{array}%
\end{equation*}%
represents a nonlinear connection on $E^{\ast }$. This is called the \textbf{%
canonical nonlinear connection of the multi-time Hamilton space }$%
MH_{m}^{n}=(E^{\ast },H).$
\end{theorem}

Taking into account the Theorem \ref{thchar} and using the \textit{%
generalized spatial Christoffel symbols} of the d-tensor $g_{ij}$, which are
given by%
\begin{equation*}
\Gamma _{ij}^{k}=\frac{g^{kl}}{2}\left( \frac{\partial g_{li}}{\partial x^{j}%
}+\frac{\partial g_{lj}}{\partial x^{i}}-\frac{\partial g_{ij}}{\partial
x^{l}}\right) ,
\end{equation*}%
we immediately obtain the following geometrical result:

\begin{corollary}
For $m=\dim \mathcal{T}\geq 2,$ the canonical nonlinear connection $N$ of a
multi-time Hamilton space $MH_{m}^{n}=(E^{\ast },H),$ whose Hamiltonian
function is given by (\ref{NEDTH}), has the components%
\begin{equation*}
\underset{1}{N}\text{{}}_{(i)b}^{(a)}=\chi _{bc}^{a}p_{i}^{c},\qquad\underset%
{2}{N}\text{{}}_{(i)j}^{(a)}=-\Gamma _{ij}^{k}p_{k}^{a}+T_{(i)j}^{(a)},
\end{equation*}%
where 
\begin{equation}
T_{(i)j}^{(a)}=\dfrac{h^{ab}}{4}\left( U_{ib\bullet j}+U_{jb\mathbf{\bullet }%
i}\right) ,  \label{aux_N2}
\end{equation}%
and%
\begin{equation*}
U_{ib}=g_{ik}U_{(b)}^{(k)},\qquad U_{kb\mathbf{\bullet }r}=\dfrac{\partial
U_{kb}}{\partial x^{r}}-U_{sb}\Gamma _{kr}^{s}.
\end{equation*}
\end{corollary}

The canonical nonlinear connection $N=\left( \underset{1}{N}\text{{}}%
_{(i)b}^{(a)},\text{ }\underset{2}{N}\text{{}}_{(i)j}^{(a)}\right) $ on $%
E^{\ast }\,$allows us the construction of the \textit{adapted bases}%
\begin{equation*}
\left\{ \dfrac{\delta }{\delta t^{a}},\dfrac{\delta }{\delta x^{i}},\dfrac{%
\partial }{\partial p_{i}^{a}}\right\} \subset \chi \left( E^{\ast }\right)
,\qquad \left\{ dt^{a},dx^{i},\delta p_{i}^{a}\right\} \subset \chi ^{\ast
}\left( E^{\ast }\right) ,
\end{equation*}%
where 
\begin{equation}
\begin{array}{c}
\dfrac{\delta }{\delta t^{a}}=\dfrac{\partial }{\partial t^{a}}-\underset{1}{%
N}\underset{}{\overset{\left( b\right) }{_{\left( j\right) a}}}\dfrac{%
\partial }{\partial p_{j}^{b}},\medskip \qquad \dfrac{\delta }{\delta x^{i}}=%
\dfrac{\partial }{\partial x^{i}}-\underset{2}{N}\underset{}{\overset{\left(
b\right) }{_{\left( j\right) i}}}\dfrac{\partial }{\partial p_{j}^{b}}, \\ 
\delta p_{i}^{a}=dp_{i}^{a}+\underset{1}{N}\overset{(a)}{\underset{}{%
_{\left( i\right) b}}}dt^{b}+\underset{2}{N}\overset{(a)}{\underset{}{%
_{\left( i\right) j}}}dx^{j}.%
\end{array}
\label{bazele_expl}
\end{equation}

The main result of the metrical multi-time Hamilton geometry is the Theorem
of existence of the Cartan canonical $h$-normal $N$-linear connection $%
C\Gamma \left( N\right) $ (see \cite{Oana+Neag-2}) which allows the
subsequent development of our metrical multi-time Hamilton theory of
physical fields.

\begin{theorem}[the Cartan canonical $N$-linear connection]
On the metrical multi-time Hamilton space $MH_{m}^{n}$ = $(J^{1\ast }(%
\mathcal{T},M),H),$ endowed with the canonical nonlinear connection $N,$
there exists an unique $h$-normal $N$-linear connection%
\begin{equation*}
C\Gamma (N)=\left( \chi _{bc}^{a},\text{ }A_{jc}^{i},\text{ }H_{jk}^{i},%
\text{ }C_{j\left( c\right) }^{i\left( k\right) }\right) ,
\end{equation*}%
having the metrical properties:\medskip

\emph{(i)} $g_{ij|k}=0,\quad g^{ij}|_{(c)}^{(k)}=0$,\medskip

\emph{(ii)} ${A_{jc}^{i}={\dfrac{g^{il}}{2}}{\dfrac{\delta g_{lj}}{\delta
t^{c}}},\quad H_{jk}^{i}=H_{kj}^{i},\quad C_{j(c)}^{i(k)}=C_{j(c)}^{k(i)},}$%
\medskip \newline
where \textquotedblright $_{/a}$\textquotedblright $,$ \textquotedblright $%
_{|k}$\textquotedblright\ and \textquotedblright $|_{(c)}^{(k)}$%
\textquotedblright\ represent the local covariant derivatives of $C\Gamma
(N) $. Moreover, the adapted local coefficients ${H_{jk}^{i}}$ and ${%
C_{j(c)}^{i(k)}}$ of the \textbf{Cartan canonical connection} $C\Gamma (N)$
have the expressions%
\begin{equation*}
H{_{jk}^{i}=}\dfrac{g^{ir}}{2}\left( {\dfrac{\delta g_{jr}}{\delta x^{k}}}+{%
\dfrac{\delta g_{kr}}{\delta x^{j}}}-{\dfrac{\delta g_{jk}}{\delta x^{r}}}%
\right) {,\quad }C_{i(c)}^{j(k)}{=-{\dfrac{g_{ir}}{2}}\left( {\dfrac{%
\partial g^{jr}}{\partial p_{k}^{c}}}+{\dfrac{\partial g^{kr}}{\partial
p_{j}^{c}}}-{\dfrac{\partial g^{jk}}{\partial p_{r}^{c}}}\right) }.
\end{equation*}
\end{theorem}

\begin{remark}
\emph{(i)} The Cartan canonical connection $C\Gamma (N)$ of the multi-time
Hamilton space $MH_{m}^{n}$ verifies also the metrical properties%
\begin{equation*}
h_{ab/c}=h_{ab|k}=h_{ab}|_{(c)}^{(k)}=0,\quad g_{ij/c}=0.
\end{equation*}

\emph{(ii)} In the case $m=\dim \mathcal{T}\geq 2$, the adapted coefficients
of the Cartan canonical connection $C\Gamma (N)$ of the multi-time Hamilton
space $MH_{m}^{n}$ reduce to%
\begin{equation}
A_{bc}^{a}=\chi _{bc}^{a},\quad A_{jc}^{i}={{\dfrac{g^{il}}{2}}{\dfrac{%
\partial g_{lj}}{\partial t^{c}}}},\quad H_{jk}^{i}=\Gamma _{jk}^{i},\quad {%
C_{j(c)}^{i(k)}}=0.  \label{Cartan-local-coeff}
\end{equation}
\end{remark}

Applying the formulas that determine the local d-torsions and d-curvatures
of an $h$-normal $N$-linear connection $D\Gamma (N)$ (see \cite{Oana+Neag})
to the Cartan canonical connection $C\Gamma (N)$, we obtain (see \cite%
{Oana+Neag-2}):

\begin{theorem}
The torsion tensor $\mathbb{T}$ of the Cartan canonical connection $C\Gamma
(N)$ of the multi-time Hamilton space $MH_{m}^{n}$ is determined by the
adapted local d-components%
\begin{equation*}
\begin{tabular}{|l|l|l|l|l|l|}
\hline
& $h_{\mathcal{T}}$ & \multicolumn{2}{|l|}{$h_{M}$} & \multicolumn{2}{|l|}{$%
v $} \\ \hline
& $m\geq 1$ & $m=1$ & $m\geq 2$ & $m=1$ & $m\geq 2$ \\ \hline
$h_{\mathcal{T}}h_{\mathcal{T}}$ & $0$ & $0$ & $0$ & $0$ & $R_{(r)ab}^{(f)}$
\\ \hline
$h_{M}h_{\mathcal{T}}$ & $0$ & $T_{1j}^{r}$ & $T_{aj}^{r}$ & $%
R_{(r)1j}^{(1)} $ & $R_{(r)aj}^{(f)}$ \\ \hline
$vh_{\mathcal{T}}$ & $0$ & $0$ & $0$ & $P_{(r)1(1)}^{(1)\;\;(j)}$ & $%
P_{(r)a(b)}^{(f)\;\;(j)}$ \\ \hline
$h_{M}h_{M}$ & $0$ & $0$ & $0$ & $R_{(r)ij}^{(1)}$ & $R_{(r)ij}^{(f)}$ \\ 
\hline
$vh_{M}$ & $0$ & $P_{i(1)}^{r(j)}$ & $0$ & $P_{(r)i(1)}^{(1)\;(j)}$ & $0$ \\ 
\hline
$vv$ & $0$ & $0$ & $0$ & $0$ & $0$ \\ \hline
\end{tabular}%
\end{equation*}%
where

\emph{(i)} for $m=\dim \mathcal{T}=1,$ we have%
\begin{equation*}
\begin{array}{c}
T_{1j}^{r}=-A_{j1}^{r},\quad P_{i(1)}^{r(j)}=C_{i(1)}^{r(j)},\quad
P_{(r)1(1)}^{(1)\;\ (j)}=\dfrac{\partial \underset{1}{N}\overset{\left(
1\right) }{_{\left( r\right) 1}}}{\partial p_{j}^{1}}+A_{r1}^{j}-\delta
_{r}^{j}\chi _{11}^{1},\medskip \\ 
{P_{(r)i(1)}^{(1)\;(j)}={\dfrac{\partial \underset{2}{N}\overset{\left(
1\right) }{_{\left( r\right) i}}}{\partial p_{j}^{1}}+}}H_{ri}^{j}{,}\quad
R_{(r)1j}^{(1)}={{\dfrac{\delta \underset{1}{N}\overset{\left( 1\right) }{%
_{\left( r\right) 1}}}{\delta x^{j}}}-{\dfrac{\delta \underset{2}{N}\overset{%
\left( 1\right) }{_{\left( r\right) j}}}{\delta t},}}\medskip \\ 
{{R_{(r)ij}^{(1)}}={{\dfrac{\delta \underset{2}{N}\overset{\left( 1\right) }{%
_{\left( r\right) i}}}{\delta x^{j}}}-{\dfrac{\delta \underset{2}{N}\overset{%
\left( 1\right) }{_{\left( r\right) j}}}{\delta x^{i}}}};}%
\end{array}%
\end{equation*}

\emph{(ii)} for $m=\dim \mathcal{T}\geq 2,$ using the equality (\ref{aux_N2}%
) and the notations%
\begin{equation*}
\begin{array}{l}
\medskip \chi _{fab}^{c}=\dfrac{\partial \chi _{fa}^{c}}{\partial t^{b}}-%
\dfrac{\partial \chi _{fb}^{c}}{\partial t^{a}}+\chi _{fa}^{d}\chi
_{db}^{c}-\chi _{fb}^{d}\chi _{da}^{c}, \\ 
\mathfrak{R}_{kij}^{r}=\dfrac{\partial \Gamma _{ki}^{r}}{\partial x^{j}}-%
\dfrac{\partial \Gamma _{kj}^{r}}{\partial x^{i}}+\Gamma _{ki}^{p}\Gamma
_{pj}^{r}-\Gamma _{kj}^{p}\Gamma _{pi}^{r},%
\end{array}%
\end{equation*}%
we have%
\begin{equation*}
\begin{array}{l}
\medskip T_{aj}^{r}=-A_{ja}^{r},\quad P_{(r)a(b)}^{(f)\;\;(j)}=\delta
_{b}^{f}A_{ra}^{j},\quad R_{(r)ab}^{(f)}=\chi _{gab}^{f}p_{r}^{g}, \\ 
\medskip R_{(r)aj}^{(f)}=-{\dfrac{\partial \underset{2}{N}\overset{\left(
f\right) }{_{\left( r\right) j}}}{\partial t^{a}}}-\chi _{ca}^{f}T{%
_{(r)j}^{(c)}}, \\ 
R_{(r)ij}^{(f)}=-\mathfrak{R}_{rij}^{k}p_{k}^{f}+\left[
T_{(r)i|j}^{(f)}-T_{(r)j|i}^{(f)}\right] .%
\end{array}%
\end{equation*}
\end{theorem}

\begin{theorem}
The curvature tensor $\mathbb{R}$ of the Cartan connection $C\Gamma (N)$ of
the multi-time Hamilton space $MH_{m}^{n}$ is determined by the following
adapted local curvature d-tensors:%
\begin{equation*}
\begin{tabular}{|l|l|l|l|l|l|}
\hline
& $h_{\mathcal{T}}$ & \multicolumn{2}{|l|}{$h_{M}$} & \multicolumn{2}{|l|}{$%
v $} \\ \hline
& $m\geq 1$ & $m=1$ & $m\geq 2$ & $m=1$ & $m\geq 2$ \\ \hline
$h_{\mathcal{T}}h_{\mathcal{T}}$ & $\chi _{abc}^{d}$ & $0$ & $R_{ibc}^{l}$ & 
$0$ & $-R_{(l)(a)bc}^{(d)(i)}$ \\ \hline
$h_{M}h_{\mathcal{T}}$ & $0$ & $R_{i1k}^{l}$ & $R_{ibk}^{l}$ & $%
-R_{(i)(1)1k}^{(1)(l)}=-R_{i1k}^{l}$ & $-R_{(l)(a)bk}^{(d)(i)}$ \\ \hline
$vh_{\mathcal{T}}$ & $0$ & $P_{i1(1)}^{l\;\;(k)}$ & $0$ & $%
-P_{(i)(1)1(1)}^{(1)(l)\;(k)}=-P_{i1(1)}^{l\;\;(k)}$ & $0$ \\ \hline
$h_{M}h_{M}$ & $0$ & $R_{ijk}^{l}$ & $\mathfrak{R}_{ijk}^{l}$ & $%
-R_{(i)(1)jk}^{(1)(l)}=-R_{ijk}^{l}$ & $-R_{(l)(a)jk}^{(d)(i)}$ \\ \hline
$vh_{M}$ & $0$ & $P_{ij(1)}^{l\;(k)}$ & $0$ & $-P_{(i)(1)j(1)}^{(1)(l)%
\;(k)}=-P_{ij(1)}^{l\;(k)}$ & $0$ \\ \hline
$vv$ & $0$ & $S_{i(1)(1)}^{l(j)(k)}$ & $0$ & $%
-S_{(i)(1)(1)(1)}^{(1)(l)(j)(k)}=-S_{i(1)(1)}^{l(j)(k)}$ & $0$ \\ \hline
\end{tabular}%
\end{equation*}%
where, for $m=\dim \mathcal{T}\geq 2$, we have the relations%
\begin{equation*}
-R_{(l)(a)bc}^{(d)(i)}=\delta _{l}^{i}\chi _{abc}^{d}-\delta
_{a}^{d}R_{lbc}^{i},\quad -R_{(l)(a)bk}^{(d)(i)}=-\delta
_{a}^{d}R_{lbk}^{i},\quad -R_{(i)(a)jk}^{(d)(l)}=-\delta _{a}^{d}\mathfrak{R}%
_{ijk}^{l},
\end{equation*}%
and, generally, the following formulas are true:\bigskip

\emph{(i)} for $m=\dim \mathcal{T}=1,$ we have%
\begin{equation*}
\begin{array}{l}
\chi _{111}^{1}=0,\medskip \\ 
R_{i1k}^{l}=\dfrac{\delta A_{i1}^{l}}{\delta x^{k}}-\dfrac{\delta H_{ik}^{l}%
}{\delta t}+A_{i1}^{r}H_{rk}^{l}-H_{ik}^{r}A_{r1}^{l}+C_{i\left( 1\right)
}^{l\left( r\right) }R_{\left( r\right) 1k}^{\left( 1\right) },\medskip \\ 
R_{ijk}^{l}=\dfrac{\delta H_{ij}^{l}}{\delta x^{k}}-\dfrac{\delta H_{ik}^{l}%
}{\delta x^{j}}+H_{ij}^{r}H_{rk}^{l}-H_{ik}^{r}H_{rj}^{l}+C_{i\left(
1\right) }^{l\left( r\right) }R_{\left( r\right) jk}^{\left( 1\right)
},\medskip \\ 
P_{i1\left( 1\right) }^{l\ \left( k\right) }=\dfrac{\partial A_{i1}^{l}}{%
\partial p_{k}^{1}}-C_{i\left( 1\right) /1}^{l\left( k\right) }+C_{i\left(
1\right) }^{l\left( r\right) }P_{\left( r\right) 1\left( 1\right) }^{\left(
1\right) \ \left( k\right) },\medskip \\ 
P_{ij\left( 1\right) }^{l\ \left( k\right) }=\dfrac{\partial H_{ij}^{l}}{%
\partial p_{k}^{1}}-C_{i\left( 1\right) |j}^{l\left( k\right) }+C_{i\left(
1\right) }^{l\left( r\right) }P_{\left( r\right) j\left( 1\right) }^{\left(
1\right) \ \left( k\right) },\medskip \\ 
S_{i(1)(1)}^{l(j)(k)}=\dfrac{\partial C_{i(1)}^{l(j)}}{\partial p_{k}^{1}}-%
\dfrac{\partial C_{i(1)}^{l(k)}}{\partial p_{j}^{1}}%
+C_{i(1)}^{r(j)}C_{r(1)}^{l(k)}-C_{i(1)}^{r(k)}C_{r(1)}^{l(j)};%
\end{array}%
\end{equation*}

\emph{(ii)} for $m=\dim \mathcal{T}\geq 2,$ we have%
\begin{equation*}
\begin{array}{l}
\chi _{abc}^{d}=\dfrac{\partial \chi _{ab}^{d}}{\partial t^{c}}-\dfrac{%
\partial \chi _{ac}^{d}}{\partial t^{b}}+\chi _{ab}^{f}\chi _{fc}^{d}-\chi
_{ac}^{f}\chi _{fb}^{d},\medskip \\ 
R_{ibc}^{l}=\dfrac{\partial A_{ib}^{l}}{\partial t^{c}}-\dfrac{\partial
A_{ic}^{l}}{\partial t^{b}}+A_{ib}^{r}A_{rc}^{l}-A_{ic}^{r}A_{rb}^{l},%
\medskip \\ 
R_{ibk}^{l}=\dfrac{\partial A_{ib}^{l}}{\partial x^{k}}-\dfrac{\partial
\Gamma _{ik}^{l}}{\partial t^{b}}+A_{ib}^{r}\Gamma _{rk}^{l}-\Gamma
_{ik}^{r}A_{rb}^{l}{,}\medskip \\ 
\mathfrak{R}_{ijk}^{l}=\dfrac{\partial \Gamma _{ij}^{l}}{\partial x^{k}}-%
\dfrac{\partial \Gamma _{ik}^{l}}{\partial x^{j}}+\Gamma _{ij}^{r}\Gamma
_{rk}^{l}-\Gamma _{ik}^{r}\Gamma _{rj}^{l}.%
\end{array}%
\end{equation*}
\end{theorem}

In the next Sections, following the physical and geometrical ideas of the
already classical Lagrangian geometry of physical fields (see \cite%
{Miro+Anas}, \cite{Miro+Hrim+Shim+Saba} and \cite{Neag1}), we construct a
possible multi-time Hamiltonian approach of the electromagnetic and
gravitational physical fields, which is characterized by some natural
geometrical Maxwell-like and Einstein-like equations. To reach this aim, we
consider a multi-time Hamilton space $MH_{m}^{n}=\left( J^{1\ast }\left(
T,M\right) ,H\right) $ endowed with its canonical nonlinear connection%
\begin{equation*}
N=\left( \underset{1}{N}\text{{}}_{(i)b}^{(a)},\text{ }\underset{2}{N}\text{%
{}}_{(i)j}^{(a)}\right) ,
\end{equation*}%
and we also consider the Cartan canonical connection of the space $%
MH_{m}^{n} $, which is locally expressed by%
\begin{equation*}
C\Gamma (N)=\left( \chi _{bc}^{a},\text{ }A_{jc}^{i},\text{ }H_{jk}^{i},%
\text{ }C_{j\left( c\right) }^{i\left( k\right) }\right) .
\end{equation*}

\section{Multi-time Hamilton electromagnetism. Geometrical Maxwell-like
equations}

Let us consider the \textit{canonical Liouville-Hamilton d-tensor field of
polymomenta}%
\begin{equation*}
\mathbb{C}^{\ast }{=p_{i}^{a}{\dfrac{\partial }{\partial p_{i}^{a}},}}
\end{equation*}%
together with the fundamental vertical metrical d-tensor $%
G_{(a)(b)}^{(i)(j)} $ of the multi-time Hamilton spaces $MH_{m}^{n}$. These
geometrical objects allow us to construct the \textit{metrical deflection
d-tensors}%
\begin{equation*}
\begin{array}{c}
\Delta _{(a)b}^{(i)}=G_{(a)(c)}^{(i)(k)}\Delta
_{(k)b}^{(c)}=p_{(a)/b}^{(i)},\medskip \quad \Delta
_{(a)j}^{(i)}=G_{(a)(c)}^{(i)(k)}\Delta _{(k)j}^{(c)}=p_{(a)|j}^{(i)}, \\ 
\vartheta _{(a)(b)}^{(i)(j)}=G_{(a)(c)}^{(i)(k)}\vartheta
_{(k)(b)}^{(c)(j)}=p_{(a)}^{(i)}|_{(b)}^{(j)},%
\end{array}%
\end{equation*}%
where ${p_{\left( a\right) }^{\left( i\right) }=}G_{(a)(c)}^{(i)(k)}{%
p_{k}^{c}}$ and "$_{/b}$", "$_{|j}$"\ and "$|_{(b)}^{(j)}$" are the local
covariant derivatives induced by the Cartan connection $C\Gamma (N)$.

Taking into account the expressions of the local covariant derivatives of
the Cartan connection $C\Gamma (N)$ (see the paper \cite{Oana+Neag}), by a
direct calculation, we obtain

\begin{proposition}
\label{defl-metr-NO}The metrical deflection d-tensors of the multi-time
Hamilton spaces $MH_{m}^{n}$ have the expresions:\medskip

\emph{(i)} for $m=\dim \mathcal{T=}1$, we have%
\begin{equation}
\begin{array}{c}
\Delta _{(1)1}^{(i)}=-h_{11}g^{ik}A_{k1}^{r}p_{r}^{1},\medskip \quad\Delta
_{(1)j}^{(i)}=h_{11}g^{ik}\left[ -\underset{2}{N}\overset{\left( 1\right) }{%
_{\left( k\right) j}}-H_{kj}^{r}p_{r}^{1}\right] , \\ 
\vartheta _{(1)(1)}^{(i)(j)}=h_{11}g^{ij}-h_{11}g^{ik}C_{k\left( 1\right)
}^{r\left( j\right) }p_{r}^{1};%
\end{array}
\label{defl-metr-1}
\end{equation}

\emph{(ii)} for $m=\dim \mathcal{T}\geq 2$, we have%
\begin{equation}
\begin{array}{c}
\Delta _{(a)b}^{(i)}=-h_{ac}g^{ik}A_{kb}^{r}p_{r}^{c},\medskip \quad\Delta
_{(a)j}^{(i)}=-\dfrac{g^{ik}}{4}\left( U_{ka\bullet j}+U_{ja\bullet
k}\right) , \\ 
\vartheta _{(a)(b)}^{(i)(j)}=h_{ab}g^{ij}.%
\end{array}
\label{defl-metr-2}
\end{equation}
\end{proposition}

In order to construct our metrical multi-time Hamiltonian theory of
electromagnetism, we introduce the following concept:

\begin{definition}
The distinguished 2-form on $J^{1\ast }\left( \mathcal{T},M\right) ,$
locally defined by 
\begin{equation}
\mathbb{F}=F_{(a)j}^{(i)}\delta p_{i}^{a}\wedge
dx^{j}+f_{(a)(b)}^{(i)(j)}\delta p_{i}^{a}\wedge \delta p_{j}^{b},
\label{electromagnetic-NO}
\end{equation}%
where%
\begin{equation}
F_{(a)j}^{(i)}={{\dfrac{1}{2}}\left[ \Delta _{(a)j}^{(i)}-\Delta
_{(a)i}^{(j)}\right] },\qquad f_{(a)(b)}^{(i)(j)}={{\dfrac{1}{2}}\left[
\vartheta _{(a)(b)}^{(i)(j)}-\vartheta _{(a)(b)}^{(j)(i)}\right] },
\label{(F_ij )si (f_ij)-NO}
\end{equation}%
is called the \textbf{multi-time electromagnetic field of the metrical
multi-time Hamilton space }$MH_{m}^{n}$.
\end{definition}

By a straightforward calculation, the Proposition \ref{defl-metr-NO} implies

\begin{proposition}
The components $F_{(a)j}^{(i)}$ and $f_{(a)(b)}^{(i)(j)}$ of the multi-time
electromagnetic field $\mathbb{F}$, associated to the multi-time Hamilton
space $MH_{m}^{n}$, have the following expressions:\medskip

\emph{(i)} in the case $m=\dim \mathcal{T}=1$, we have%
\begin{equation*}
F_{(1)j}^{(i)}={\frac{h^{11}}{2}}\left[ g^{jk}\underset{2}{N}\overset{\left(
1\right) }{_{\left( k\right) i}}-g^{ik}\underset{2}{N}\overset{\left(
1\right) }{_{\left( k\right) j}}+\left(
g^{jk}H_{ki}^{r}-g^{ik}H_{kj}^{r}\right) p_{r}^{1}\right] ,\quad
f_{(1)(1)}^{(i)(j)}=0;
\end{equation*}

\emph{(ii)} in the case $m=\dim \mathcal{T}\geq 2$, we have%
\begin{equation*}
F_{(a)j}^{(i)}=\frac{1}{8}\left[ g^{jk}U_{ka\bullet i}-g^{ik}U_{ka\bullet
j}+g^{jk}U_{ia\bullet k}-g^{ik}U_{ja\bullet k}\right] ,\quad
f_{(a)(b)}^{(i)(j)}=0.
\end{equation*}
\end{proposition}

The main result of our abstract geometrical Hamilton multi-time
electromagnetism is given by

\begin{theorem}
The electromagnetic components $F_{(a)j}^{(i)}$ of the multi-time Hamilton
space $MH_{m}^{n}\ $are governed by the following \textbf{geometrical
Maxwell-like equations}: \medskip

\emph{(i)} for $m=\dim \mathcal{T}=1$, we have%
\begin{equation*}
\left\{ 
\begin{array}{l}
\medskip {F_{(1)k/1}^{(i)}={\dfrac{1}{2}}\mathcal{A}_{\{i,k\}}\left\{ \Delta
_{(1)1|k}^{(i)}+\Delta _{(1)r}^{(i)}T_{1k}^{r}+\vartheta
_{(1)(1)}^{(i)(r)}R_{(r)1k}^{(1)}+R_{r1k}^{i}p_{\left( 1\right) }^{\left(
r\right) }\right\} } \\ 
\medskip {\sum_{\{i,j,k\}}F_{(1)j|k}^{(i)}=-{\dfrac{1}{2}}\sum_{\{i,j,k\}}}%
\left\{ {\vartheta _{(1)(1)}^{(i)(r)}R_{(r)jk}^{(1)}+}R_{rjk}^{i}p_{\left(
1\right) }^{\left( r\right) }\right\} \\ 
{F_{(1)j}^{(i)}|_{(1)}^{(k)}={\dfrac{1}{2}}\mathcal{A}_{\{i,j\}}}\left\{ {%
\vartheta _{(1)(1)|j}^{(i)(k)}-}P_{rj(1)}^{i\;\;(k)}p_{\left( 1\right)
}^{\left( r\right) }-\Delta _{(1)r}^{(i)}C_{j\left( 1\right) }^{r\left(
k\right) }-{\vartheta _{(1)(1)}^{(i)(r)}}P_{(r)j(1)}^{(1)\;\;(k)}\right\} {;}%
\end{array}%
\right.
\end{equation*}

\emph{(ii)} for $m=\dim \mathcal{T}\geq 2$, we have%
\begin{equation*}
\left\{ 
\begin{array}{l}
\medskip {F_{(a)k/b}^{(i)}={\dfrac{1}{2}}\mathcal{A}_{\{i,k\}}\left\{ \Delta
_{(a)b|k}^{(i)}+\Delta _{(a)r}^{(i)}T_{bk}^{r}+\vartheta
_{(a)(f)}^{(i)(r)}R_{(r)bk}^{(f)}+R_{rbk}^{i}p_{\left( a\right) }^{\left(
r\right) }\right\} } \\ 
\medskip {\sum_{\{i,j,k\}}F_{(a)j|k}^{(i)}=-{\dfrac{1}{2}}\sum_{\{i,j,k\}}}%
\left\{ {\vartheta _{(a)(f)}^{(i)(r)}R_{(r)jk}^{(f)}+}\mathfrak{R}%
_{rjk}^{i}p_{\left( a\right) }^{\left( r\right) }\right\} \\ 
\sum_{\{i,j,k\}}{F_{(a)j}^{(i)}|_{(c)}^{(k)}={0}},%
\end{array}%
\right.
\end{equation*}%
where ${\mathcal{A}_{\{i,j\}}}$ means an alternate sum, ${\sum_{\{i,j,k\}}}$
means a cyclic sum, and we used the notations $%
p_{(1)}^{(i)}=G_{(1)(1)}^{(i)(j)}p_{j}^{1}$ and $%
p_{(a)}^{(i)}=G_{(a)(b)}^{(i)(j)}p_{j}^{b}$.
\end{theorem}

\begin{proof}
The general Ricci identities (see \cite{Oana+Neag} and \cite{Oana+Neag-2})
applied to $g^{ij}$ lead us to the equalities:%
\begin{equation}
\begin{array}{c}
g^{ir}R_{rbk}^{j}+g^{jr}R_{rbk}^{i}=0,\medskip\qquad
g^{ir}R_{rkl}^{j}+g^{jr}R_{rkl}^{i}=0, \\ 
g^{ir}P_{rk(c)}^{j\text{ \ }(l)}+g^{jr}P_{rk(c)}^{i\text{ \ }(l)}=0.%
\end{array}
\label{anti-simetrie-curvature}
\end{equation}%
Let us consider the following non-metrical deflection d-tensor identities (%
\cite{Oana+Neag}):\medskip

$(d_{1})\;\;\Delta _{(p)b|k}^{(d)}-\Delta
_{(p)k/b}^{(d)}=p_{r}^{d}R_{pbk}^{r}-\Delta
_{(p)r}^{(d)}T_{bk}^{r}-\vartheta _{(p)(f)}^{(d)(r)}R_{(r)bk}^{(f)},$\medskip

$(d_{2})\;\;\Delta _{(p)j|k}^{(d)}-\Delta
_{(p)k|j}^{(d)}=p_{r}^{d}R_{pjk}^{r}-\vartheta
_{(p)(f)}^{(d)(r)}R_{(r)jk}^{(f)},$\medskip

$(d_{3})\;\;\Delta _{(p)j}^{(d)}|_{(c)}^{(k)}-\vartheta
_{(p)(c)|j}^{(d)(k)}=p_{r}^{d}P_{pj(c)}^{r\;\;(k)}-\Delta
_{(p)r}^{(d)}C_{j(c)}^{r(k)}-\vartheta _{(p)(f)}^{(d)(r)}P_{(r)j(c)}^{(f)\;\
(k)}$,\medskip \newline
where $\Delta _{(i)b}^{(a)}=p_{i/b}^{a}$, $\Delta _{(i)j}^{(a)}=p_{i|j}^{a}$%
, $\vartheta _{(i)(b)}^{(a)(j)}=p_{i}^{a}|_{(b)}^{(j)}$.\medskip

Contracting the above deflection d-tensor identities with the fundamental
vertical metrical d-tensor $G_{(a)(d)}^{(i)(p)}$, and using the equalities (%
\ref{anti-simetrie-curvature}), we obtain the following \textit{metrical
deflection d-tensor identities:}\medskip

$(d_{1}^{\prime })\;\;\Delta _{(a)b|k}^{(i)}-\Delta
_{(a)k/b}^{(i)}=-p_{(a)}^{(r)}R_{rbk}^{i}-\Delta
_{(a)r}^{(i)}T_{bk}^{r}-\vartheta _{(a)(f)}^{(i)(r)}R_{(r)bk}^{(f)},$\medskip

$(d_{2}^{\prime })\;\;\Delta _{(a)j|k}^{(i)}-\Delta
_{(a)k|j}^{(i)}=-p_{(a)}^{(r)}R_{rjk}^{i}-\vartheta
_{(a)(f)}^{(i)(r)}R_{(r)jk}^{(f)},$\medskip

$(d_{3}^{\prime })\;\;\Delta _{(a)j}^{(i)}|_{(c)}^{(k)}-\vartheta
_{(a)(c)|j}^{(i)(k)}=-p_{(a)}^{(r)}P_{rj(c)}^{i\;\;(k)}-\Delta
_{(a)r}^{(i)}C_{j(c)}^{r(k)}-\vartheta _{(a)(f)}^{(i)(r)}P_{(r)j(c)}^{(f)\;\
(k)}$.\medskip

To obtain the first (respectively, the third) geometrical Maxwell-like
equation, we permute the indices $i$ and $k$ in the identity $(d_{1}^{\prime
})$ (respectively, the indices $i$ and $j$ in the identity $(d_{3}^{\prime
}) $), and we subtract this new identity from the initial one. For $m=\dim 
\mathcal{T}\geq 2$ we find the following new identity:%
\begin{equation*}
{F_{(a)j}^{(i)}|_{(c)}^{(k)}={\dfrac{1}{2}}}\left[ \vartheta
_{(a)(c)|j}^{(i)(k)}{-}\vartheta _{(a)(c)|i}^{(j)(k)}\right] .
\end{equation*}%
Consequently, doing a cyclic sum by $\{i,j,k\}$ for $m\geq 2$, we obtain
what we were looking for.

Doing a cyclic sum after the indices $\{i,j,k\}$ in the identity $%
(d_{2}^{\prime })$, it follows the second geometrical Maxwell-like equation.
\end{proof}

\section{Multi-time Hamilton gravitational field. Geometrical Einstein-like
equations}

Let us consider that $h=(h_{ab}(t))$ is a fixed semi-Riemannian metric on
the temporal manifold $\mathcal{T}$ and let 
\begin{equation*}
N=\left( \underset{1}{N}\text{{}}_{(i)b}^{(a)},\text{ }\underset{2}{N}\text{%
{}}_{(i)j}^{(a)}\right)
\end{equation*}%
be an \textquotedblright \textit{a priori}\textquotedblright\ given
nonlinear connection on the dual $1$-jet space $J^{1\ast }(\mathcal{T},M)$.
Let 
\begin{equation*}
\delta p_{i}^{a}=dp_{i}^{a}+\underset{1}{N}\text{{}}_{(i)b}^{(a)}dt^{b}+%
\underset{2}{N}\text{{}}_{(i)j}^{(a)}dx^{j}
\end{equation*}%
be the vertical distinguished 1-forms adapted to the nonlinear connection $N$%
.

An essential element in the development of our abstract geometrical
multi-time Hamilton gravitational theory is given by the following
definition:

\begin{definition}
From an abstract physical point of view, an adapted metrical d-tensor $%
\mathbb{G}$ on the dual $1$-jet space $E^{\ast }=J^{1\ast }(\mathcal{T},M),$
locally expressed by 
\begin{equation*}
\mathbb{G}=h_{ab}dt^{a}\otimes dt^{b}+g_{ij}dx^{i}\otimes
dx^{j}+h_{ab}g^{ij}\delta p_{i}^{a}\otimes \delta p_{j}^{b},
\end{equation*}%
where $g_{ij}=g_{ij}(t^{c},x^{k},x_{c}^{k})$ is a symmetric d-tensor field
of rank $n=\dim M$ having a constant signature on $E^{\ast }=J^{1\ast }(%
\mathcal{T},M),$ is called a \textbf{multi-time gravitational }$h$\textbf{%
-potential}\textit{\ on }$E^{\ast }$.
\end{definition}

Now, taking a multi-time Hamilton space $MH_{m}^{n}=(E^{\ast },H)$, via its
fundamental vertical metrical d-tensor $G_{(a)(b)}^{(i)(j)}$ (which is given
by (\ref{FVDT})) and its canonical nonlinear connection $N$, we naturally
construct a corresponding multi-time gravitational $h$-potential on $E^{\ast
}$, setting%
\begin{equation*}
\mathbb{G}=h_{ab}dt^{a}\otimes dt^{b}+g_{ij}dx^{i}\otimes
dx^{j}+h_{ab}g^{ij}\delta p_{i}^{a}\otimes \delta p_{j}^{b}.
\end{equation*}%
At the same time, let us consider that 
\begin{equation*}
C\Gamma (N)=\left( \chi
_{ab}^{c},A_{jc}^{k},H_{jk}^{i},C_{j(c)}^{i(k)}\right)
\end{equation*}%
is the Cartan canonical connection of the multi-time Hamilton space $%
MH_{m}^{n}$.\medskip

\textbf{Postulate.} \textit{We postulate that the \textbf{geometrical
Einstein-like equations}, which govern the multi-time gravitational }$h$%
\textit{-potential }$\mathbb{G}$\textit{\ of the multi-time Hamilton space }$%
MH_{m}^{n}$\textit{,\ are the abstract geometrical Einstein equations
attached to the Cartan canonical connection }$C\Gamma (N)$\textit{\ and to
the adapted metric }$\mathbb{G}$\textit{\ on }$E^{\ast }$\textit{,\ namely}%
\begin{equation}
\text{Ric}(C\Gamma )-{\frac{\text{Sc}(C\Gamma )}{2}}\mathbb{G}=\mathcal{K}%
\mathbb{T},  \label{Einstein-Global-Abstract}
\end{equation}%
\textit{where }Ric$(C\Gamma )$\textit{\ represents the \textbf{Ricci tensor}
of the Cartan connection, }Sc$(C\Gamma )$\textit{\ is the \textbf{scalar
curvature}, }$\mathcal{K}$\textit{\ is the \textbf{Einstein constant} and }$%
\mathbb{T}$\textit{\ is an intrinsic d-tensor of matter, which is called the 
\textbf{stress-energy d-tensor of polymomenta}.}\medskip

In the adapted basis of vector fields%
\begin{equation*}
(X_{A})={\left( {\frac{\delta }{\delta t^{a}}},{\frac{\delta }{\delta x^{i}}}%
,{\frac{\partial }{\partial p_{i}^{a}}}\right) ,}
\end{equation*}%
which is produced by the canonical nonlinear connection $N$ of the
multi-time Hamilton space $MH_{m}^{n}$, the curvature tensor $\mathbb{R}$ of
the Cartan canonical connection $C\Gamma (N)$ is locally expressed by 
\begin{equation*}
\mathbb{R}(X_{C},X_{B})X_{A}=\mathbf{R}_{ABC}^{D}X_{D}.
\end{equation*}%
It follows that we have 
\begin{equation*}
R_{AB}=\text{Ric}(X_{A},X_{B})=\mathbf{R}_{ABD}^{D},\qquad\text{Sc}(C\Gamma
)=G^{AB}R_{AB},
\end{equation*}%
where%
\begin{equation}
G^{AB}=\left\{ 
\begin{array}{ll}
\medskip h^{ab}, & \mbox{for}\;\;A=a,\;B=b \\ 
\medskip g^{ij}, & \mbox{for}\;\;A=i,\;B=j \\ 
\medskip h^{ab}g_{ij}, & \mbox{for}\;\;A={\QATOP{(a)}{(i)}},\;B={\QATOP{(b)}{%
(j)}} \\ 
0, & \mbox{otherwise}.%
\end{array}%
\right.  \label{fdt-Inv-ML}
\end{equation}

Taking into account, on the one hand, the form of the metrical d-tensor $%
\mathbb{G}=(G_{AB})$ of the multi-time Hamilton space $MH_{m}^{n}$, and, on
the other hand, taking into account the expressions of the local curvature
d-tensors attached to the Cartan canonical connection $C\Gamma \left(
N\right) $, by direct computations, we find

\begin{proposition}
The Ricci tensor \emph{Ric}$(C\Gamma )$ of the Cartan canonical connection $%
C\Gamma \left( N\right) $ of the multi-time Hamilton space $MH_{m}^{n}$ is
determined by the following adapted components:\medskip

\emph{(i)} for $m=\dim \mathcal{T}=1,$ we have%
\begin{equation*}
\begin{array}{ll}
R_{11}:=\chi _{11}=0, & R_{1i}=R_{1i1}^{1}=0,\medskip \\ 
R_{1(1)}^{\;\ (i)}=-P_{1(1)1}^{1(i)}=0, & R_{i1}=R_{i1r}^{r},\medskip\qquad
R_{ij}=R_{ijr}^{r}, \\ 
R_{(1)1}^{(i)}:=-P_{(1)1}^{(i)}=-P_{r1(1)}^{i\;\ (r)}, & R_{i(1)}^{%
\;(j)}:=-P_{i(1)}^{\;(j)}=-P_{ir(1)}^{r\;(j)},\medskip \\ 
R_{(1)(1)}^{(i)(j)}:=-S_{(1)(1)}^{(i)(j)}=-S_{r(1)(1)}^{i(j)(r)}, & 
R_{(1)j}^{(i)}:=-P_{(1)j}^{(i)}=-P_{rj(1)}^{i\;\ (r)};\medskip%
\end{array}%
\end{equation*}

\emph{(ii)} for $m=\dim \mathcal{T}\geq 2,$ we have%
\begin{equation*}
\begin{array}{ll}
R_{ab}:=\chi _{ab}=\chi _{abf}^{f}, & R_{ai}=R_{aif}^{f}=0,\medskip \\ 
R_{a(b)}^{\;\ (j)}=-P_{af(b)}^{f\;\ (j)}=0, & R_{ia}=R_{iar}^{r},\medskip%
\qquad R_{ij}=R_{ijr}^{r}, \\ 
R_{(a)(b)}^{(i)(j)}:=-S_{(a)(b)}^{(i)(j)}=-S_{r(a)(b)}^{i(j)(r)}=0, & 
R_{i(b)}^{\;(j)}:=-P_{i(b)}^{\;(j)}=-P_{ir(b)}^{r\;(j)}=0,\medskip \\ 
R_{(a)b}^{(i)}:=-P_{(a)b}^{(i)}=-P_{rb(a)}^{i\;\ (r)}=0, & 
R_{(a)j}^{(i)}:=-P_{(a)j}^{(i)}=-P_{rj(a)}^{i\;(r)}=0.\medskip%
\end{array}%
\end{equation*}
\end{proposition}

Using the notations $\chi =h^{ab}\chi _{ab},\;R=g^{ij}R_{ij}$ and $%
S=h^{11}g_{ij}S_{(1)(1)}^{(i)(j)}$, we obtain

\begin{corollary}
The scalar curvature \emph{Sc}$(C\Gamma )$ of the Cartan canonical
connection $C\Gamma \left( N\right) $ of the multi-time Hamilton space $%
MH_{m}^{n}$ is given by the formulas:%
\index{scalar curvature!of a multi-time Lagrange space@\textit{of a multi-time Lagrange space}|textbf}
\medskip

\emph{(i)} for $m=\dim \mathcal{T}=1,$ we have $\emph{Sc}(C\Gamma )=R-S;$

\emph{(ii)} for $m=\dim \mathcal{T}\geq 2,$ we have $\emph{Sc}(C\Gamma
)=\chi +R.$
\end{corollary}

In this context, the main result of the Hamilton geometrical multi-time
gravitational theory is offered by

\begin{theorem}
The \textbf{geometrical Einstein-like equations}, which govern the
multi-time gravitational $h$-po\-ten\-ti\-al $\mathbb{G}$ of the multi-time
Hamilton space $MH_{m}^{n}$, have the following adapted local form:\medskip

\emph{(i)} for $m=\dim \mathcal{T}=1,$ we have%
\begin{equation}
\left\{ 
\begin{array}{l}
\medskip {-{%
\dfrac{R-S}{2}}h_{11}=\mathcal{K}}\mathbb{T}{_{11}} \\ 
\medskip {R_{ij}-{\dfrac{R-S}{2}}g_{ij}=\mathcal{K}}\mathbb{T}{_{ij}} \\ 
-{S_{(1)(1)}^{(i)(j)}-{\dfrac{R-S}{2}{h}_{11}}}g^{ij}={\mathcal{K}}\mathbb{T}%
{_{(1)(1)}^{(i)(j)}}\medskip \\ 
\medskip 0=\mathbb{T}_{1i},\qquad R_{i1}=\mathcal{K}\mathbb{T}_{i1} \\ 
\medskip 0=\mathbb{T}_{1(1)}^{\text{ \ }\!(i)},\qquad-P_{i(1)}^{\;(j)}=%
\mathcal{K}\mathbb{T}_{i(1)}^{\;(j)} \\ 
-P_{(1)1}^{(i)}=\mathcal{K}\mathbb{T}_{(1)1}^{(i)},\qquad-P_{(1)j}^{(i)}=%
\mathcal{K}\mathbb{T}_{(1)j}^{(i)};%
\end{array}%
\right.  \label{Einstein1-m=1}
\end{equation}

\emph{(ii)} for $m=\dim \mathcal{T}\geq 2,$ we have%
\begin{equation}
\left\{ 
\begin{array}{l}
\medskip \chi {_{ab}-{\dfrac{\chi +R}{2}}h_{ab}=\mathcal{K}}\mathbb{T}{_{ab}}
\\ 
\medskip {R_{ij}-{\dfrac{\chi +R}{2}}g_{ij}=\mathcal{K}}\mathbb{T}{_{ij}} \\ 
\medskip {-{\dfrac{\chi +R}{2}}h_{ab}g^{ij}=\mathcal{K}}\mathbb{T}{%
_{(a)(b)}^{(i)(j)}} \\ 
\medskip 0=\mathbb{T}_{ai},\qquad R_{ia}=\mathcal{K}\mathbb{T}_{ia} \\ 
\medskip 0=\mathbb{T}_{(a)b}^{(i)},\qquad0=\mathbb{T}_{a(b)}^{\text{ \ }%
\!(j)} \\ 
0=\mathbb{T}_{i(b)}^{\;(j)},\qquad0=\mathbb{T}_{(a)j}^{(i)},%
\end{array}%
\right.  \label{Einstein1-m>2}
\end{equation}%
where $\mathbb{T}_{AB},\;A,B\in \left\{ a,\text{ }i,\text{ }{\QATOP{(a)}{(i)}%
}\right\} $ represent the adapted components of the stress-energy d-tensor
of matter $\mathbb{T}$.
\end{theorem}

\begin{remark}
In order to have the compatibility of the system of geometrical
Einstein-like equations, it is necessary the \textquotedblright a
priori\textquotedblright\ vanishing of certain adapted components of the
stress-energy d-tensor of matter $\mathbb{T}$.
\end{remark}

From a physical point of view, it is well known that in the classical
Riemannian theory of gravity, the stress-energy d-tensor of matter have to
verify some conservation laws. By a natural extension of the classical
Riemannian conservation laws, in our geometrical Hamiltonian context, we 
\underline{postulate} the following \textit{generalized conservation laws}
of the stress-energy d-tensor of polymomenta $\mathbb{T}$:%
\begin{equation*}
\mathbb{T}_{A|B}^{B}=0,\qquad \forall \;A\in \left\{ a,\text{ }i,\text{ }%
_{(i)}^{(a)}\right\} ,
\end{equation*}%
where $\mathbb{T}_{A}^{B}=G^{BD}\mathbb{T}_{DA}$. Consequently, by
straightforward computations, we obtain

\begin{theorem}
The \textbf{generalized conservation laws} of the Einstein-like equations of
the multi-time Hamilton space $MH_{m}^{n}$ are expressed by the following
formulas:\medskip

\emph{(i)} for $m=\dim \mathcal{T}=1,$ we have%
\begin{equation}
\left\{ 
\begin{array}{l}
\medskip {\left[ {\dfrac{R-S}{2}}\right]
_{/1}=R_{1|r}^{r}-P_{(r)1}^{(1)}|_{(1)}^{(r)}} \\ 
\medskip {\left[ R_{j}^{r}-{\dfrac{R-S}{2}}\delta _{j}^{r}\right]
_{|r}=P_{(r)j}^{(1)}|_{(1)}^{(r)}} \\ 
{\left[ S_{(r)(1)}^{(1)(j)}+{\dfrac{R-S}{2}}\delta _{r}^{j}\right]
\left\vert _{(1)}^{(r)}\right. =-P_{\;\;\;(1)|r}^{r(j)}},%
\end{array}%
\right.  \label{Cons-laws-m=1}
\end{equation}%
\medskip where%
\begin{equation*}
\begin{array}{lll}
R_{1}^{i}=g^{iq}R_{q1},\medskip & P_{(i)1}^{(1)}=h^{11}g_{iq}P_{(1)1}^{(q)},
& R_{j}^{i}=g^{iq}R_{qj}, \\ 
P_{(i)j}^{(1)}=h^{11}g_{iq}P_{(1)j}^{(q)}, & P_{\;%
\;(1)}^{i(j)}=g^{iq}P_{q(1)}^{\;\;(j)}\medskip & 
S_{(i)(1)}^{(1)(j)}=h^{11}g_{iq}S_{(1)(1)}^{(q)(j)};%
\end{array}%
\end{equation*}

\emph{(ii)} for $m=\dim \mathcal{T}\geq 2,$ we have%
\begin{equation}
\left\{ 
\begin{array}{l}
\medskip {\left[ \chi _{b}^{f}-{\dfrac{\chi +R}{2}}\delta _{b}^{f}\right]
_{/f}=-R_{b|r}^{r}} \\ 
{\left[ R_{j}^{r}-{\dfrac{\chi +R}{2}}\delta _{j}^{r}\right] _{|r}=0},%
\end{array}%
\right.  \label{Cons-laws-m>2}
\end{equation}%
where 
\begin{equation*}
\begin{array}{ccc}
\chi _{b}^{f}=h^{fc}\chi _{cb}, & R_{j}^{i}=g^{iq}R_{qj}, & 
R_{b}^{i}=g^{iq}R_{qb}.%
\end{array}%
\end{equation*}
\end{theorem}

\textbf{Open problem.} The authors of this paper consider that the finding
of a possible real physical meaning of the present multi-time Hamiltonian
geometrical theories of gravity and electromagnetism may be an open problem
for physicists.

\textbf{Acknowledgements.} Many thanks go to Professor Gh. Atanasiu for his
constant encouragements and useful suggestions concerning this research
topic.

Alexandru OAN\u{A} and Mircea NEAGU

University Transilvania of Bra\c{s}ov,

Department of Mathematics - Informatics,

Blvd. Iuliu Maniu, no. 50, Bra\c{s}ov 500091, Romania.

\textit{E-mails:} alexandru.oana@unitbv.ro, mircea.neagu@unitbv.ro


\begin{thebibliography}{99}
\bibitem{Asan} G.S. Asanov, \textit{Jet extension of Finslerian gauge
approach}, Fortschritte der Physik, vol. \textbf{38}, no. \textbf{8} (1990),
571-610.

\bibitem{Atan-Neag2} Gh. Atanasiu, M. Neagu, \textit{Canonical nonlinear
connections in the multi-time Hamilton geometry}, Balkan Journal of Geometry
and Its Applications, vol. \textbf{14}, no. \textbf{2} (2009), 1-12.

\bibitem{Atan-Neag0} Gh. Atanasiu, M. Neagu, \textit{Distinguished tensors
and Poisson brackets in the multi-time Hamilton geometry}, BSG Proceedings 
\textbf{16}, Geometry Balkan Press, Bucharest (2009), 12-27.

\bibitem{Atan+Neag1} Gh. Atanasiu, M. Neagu, \textit{Distinguished torsion,
curvature and deflection tensors in the multi-time Hamilton geometry},
Electronic Journal "\textsf{Differential Geometry - Dynamical Systems}",
vol. \textbf{11} (2009), 20-40.

\bibitem{Comic2} I. \v{C}omi\'{c}, \textit{Multi-time Lagrange spaces},
Electronic Journal "\textsf{Differential Geometry - Dynamical Systems}",
vol. \textbf{7} (2005), 15-33.

\bibitem{Giac+Mang+Sard1} G. Giachetta, L. Mangiarotti, G. Sardanashvily, 
\textit{Covariant Hamiltonian field theory}, http://arxiv.org/hep-th/9904062
(1999).

\bibitem{Giac+Mang+Sard2} G. Giachetta, L. Mangiarotti, G. Sardanashvily, 
\textit{Polysymplectic Hamiltonian formalism and some quantum outcomes},
http://arxiv.org/hep-th/0411005 (2004).

\bibitem{Gota+Isen+Mars+Mont} M. Gotay, J. Isenberg, J.E. Marsden, R.
Montgomery, \textit{Momentum maps and classical fields. Part I. Covariant
field theory}, http://arxiv.org/physics/9801019 (2004).

\bibitem{Gota+Isen+Mars} M. Gotay, J. Isenberg, J.E. Marsden, \textit{%
Momentum maps and classical fields. Part II. Canonical analysis of field
theories}, http://arxiv.org/math-ph/0411032 (2004).

\bibitem{Kana1} I.V. Kanatchikov, \textit{Basic structures of the covariant
canonical formalism for fields based on the De Donder-Weyl theory},
http://arxiv.org/hep-th/9410238 (1994).

\bibitem{Kana2} I.V. Kanatchikov, \textit{On quantization of field theories
in polymomentum variables}, AIP Conf. Proc., vol. \textbf{453}, Issue 
\textbf{1} (1998), 356-367.

\bibitem{Miro+Anas} R. Miron, M. Anastasiei, \textit{The Geometry of
Lagrange Spaces: Theory and Applications}, Kluwer Academic Publishers, 1994.

\bibitem{Miro+Hrim+Shim+Saba} R. Miron, D. Hrimiuc, H. Shimada, S.V. Sab\u{a}%
u, \textit{The Geometry of Hamilton and Lagrange Spaces}, Kluwer Academic
Publishers, 2001.

\bibitem{Neag1} M. Neagu, \textit{Riemann-Lagrange Geometry on }$1$\textit{%
-Jet Spaces}, Matrix Rom, Bucharest, 2005.

\bibitem{Oana+Neag} A. Oan\u{a}, M. Neagu, \textit{The local description of
the Ricci and Bianchi identities for an }$h$\textit{-normal }$N$\textit{%
-linear connection on the dual }$1$\textit{-jet space }$J^{1\ast }(\mathcal{T%
},M)$\textit{, }http://arxiv.org/math.DG/1111.4173 (2011).

\bibitem{Oana+Neag-2} A. Oan\u{a}, M. Neagu, \textit{A distinguished
Riemannian geometrization for quadratic Hamiltonians of polymomenta,\newline
}http://arxiv.org/math.DG/1112.5442 (2011).

\bibitem{Saun} D.J. Saunders, \textit{The Geometry of Jet Bundles},
Cambridge University Press, New York, London, 1989.
\end{thebibliography}
\end{document}